\documentclass[12pt]{amsart}
\usepackage{amsfonts}
\usepackage{amsmath}
\usepackage{amssymb}
\usepackage{amsthm}
\usepackage[T1]{fontenc}
\usepackage{tensor}
\usepackage{geometry}
\usepackage{mathrsfs}
\usepackage{setspace}
\usepackage{fancyhdr}
\usepackage{enumerate}
\usepackage[foot]{amsaddr}
\newtheorem{thm}{Theorem}[section]

\newtheorem{cor}[thm]{Corollary}

\newtheorem{lem}[thm]{Lemma}

\newtheorem{prop}[thm]{Proposition}
\theoremstyle{definition}
\newtheorem{deff}[thm]{Definition}

\newcommand{\N}{\mathbb N}

\newcommand{\C}{\mathbb C}
\newcommand{\R}{\mathbb R}

\newcommand{\la}{\langle}
\newcommand{\ra}{\rangle}

\newcommand{\hi}{\mathcal H}

\begin{document}

\title[]{An elementary construction of the GKSL master equation for N-level systems}

\author[Ziemke]{Matthew Ziemke}
\address{Department of Mathematics\\
Drexel University \\ 
Philadelphia, PA}
\email{matthew.j.ziemke@drexel.edu}

\begin{abstract}
The GKSL master equation for N-level systems provides a necessary and sufficient form for the generator of a quantum dynamical semigroup in the Schrodinger picture where the underlying Hilbert space is $\C^N$. In this paper we provide a detailed, self-contained, and elementary construction of the GKSL master equation for an N-level system. We also provide necessary and sufficient conditions for forms of generators of semigroups which have some, but not all, of the defining properties of quantum dynamical semigroups. We do this in such a way to illuminate how each defining property of a quantum dynamical semigroup contributes to the form of the generators.\\

\end{abstract}

\maketitle

\begin{spacing}{1.3}
\section{Introduction}
Quantum dynamical semigroups (QDSs) are used to model irreversible open quantum systems.  In the Schrodinger picture, they are semigroups of operators $(T_t)_{t \geq}$ acting on $\mathcal{S}_1(\hi)$, the space of trace-class operators on the Hilbert space $\hi$, such that each $T_t$ preserves traces, is completely positive, and the map $t \mapsto T_t(A)$ is weakly continuous for all $A \in \mathcal{S}_1(\hi)$.  In 1976, Gorini, Kossakowski, and Sudarshan in \cite{gks} gave the form of the generator of a QDS of an N-level system.  Specifically, they showed that $L$ is the generator of a QDS of an N-level system if and only if there exist $N \times N$ complex matrices $H, G_1, \dots,G_{N^2-1}$, and positive scalars $\lambda_1, \dots ,\lambda_{N^2-1}$ such that
\begin{equation}
L(A)=-i[H,A]+\frac{1}{2}\sum_{p=1}^{N^2-1} \lambda_p ([G_p,AG_p^*]+[G_pA,G_p^*]), \quad \text{ for all } A \in M_N(\C)
\end{equation}
where $H$ is self-adjoint, $tr(H)=0$, $tr(G_p)=0$ for all $p=1, \dots , N^2-1$, and $tr(G_p^*G_q)=\delta_{pq}$ for all $p,q=1, \dots ,N^2-1$.  Around the same time, Lindblad gave a similar equation in \cite{lindblad} for the general form of generators of QDSs in the Heisenberg picture in the case when the semigroup is uniformly continuous.  Today, we refer to these equations describing the generators as GKSL master equations or Lindblad equations and they are used and studied quite extensively.

In this paper we provide a detailed, self-contained, and fairly elementary proof of the result given by Gorini, Kossakowski, and Sudarshan in \cite{gks}.  While the forward direction of our proof relies substantially on the ideas given in \cite{gks}, we do deviate at times for the sake of clarity and completeness.  In this direction we also break down the construction into steps in a way which allows us to see exactly how each of the defining properties of the QDS effect the form of the generator.  Not only does this help us to understand the interplay between the semigroup and its generator but also gives forms for generators of semigroups which have some, but not all, of the properties of a QDS.  For the other direction, the proof given in \cite{gks} relies on a result the authors cite as \cite{koss}.  Here, we are able to avoid using this result. Finally, when we say that our proof is "elementary" we mean that an undergraduate student who has taken a semester of advanced linear algebra should have the necessary background to understand the proof. Although, in the first section titled "Background," we do give the general definition for quantum dynamical semigroups which involves concepts and definitions that most readers would not expect an undergraduate student to be familiar with, they are not necessary for the remaining sections which include the proof of the GKSL master equation for N-level systems.

\section{Background}
Throughout the paper, $\hi$ denotes a separable complex Hilbert space, $\mathcal{B}(\hi)$ denotes the space of bounded linear operators on $\hi$, and $\mathcal{S}_1(\hi)$ denotes the space of trace-class operators on $\hi$.  All of our inner products are linear in the first argument.

Let $n \in \N$ and let $(e_k)_{k=1}^n$ be the standard basis for $\C^n$. Define $E_{ij}=e_i^* \otimes e_j$ where $e_i^* \otimes e_j: \C^n \rightarrow \C^n$ is given by $e_i^* \otimes e_j(x)=\la x,e_i \ra e_j$ for all $x \in \C^n$.  That is, $E_{ij}$ is the $n \times n$ matrix with a 1 in the $(i,j)$th coordinate and zeros everywhere else.  Also, for $k \in \N$, $M_k(\C)$ denotes the space of all $k \times k$ matrices with complex coefficients and the space $\mathcal{B}(\hi) \otimes M_k(\C)$ is the space of all $k \times k$ matrices whose entries are elements of $\mathcal{B}(\hi)$.

Let $\mathcal{A}$ be a $^*$-subalgebra of $\mathcal{B}(\hi)$, $T:\mathcal{A} \rightarrow \mathcal{A}$ be a bounded linear operator and let $k \in \N$.  Throughout the paper, we will denote by $T^{(k)}$ the operator $T^{(k)}:\mathcal{A}\otimes M_k(\C) \rightarrow \mathcal{A} \otimes M_k(\C)$ given by $T^{(k)}(A \otimes E_{ij})=T(A) \otimes E_{ij}$ for all $A \in \mathcal{A}$ and all $i,j=1, \dots ,k$.  We begin with some definitions.

\begin{deff}
Let $\mathcal{A}$ be a $^*$-subalgebra of $\mathcal{B}(\hi)$ and let $T:\mathcal{A}\rightarrow \mathcal{A}$ be a bounded linear operator.  We say the operator $T$ is {\bf completely positive} if and only if the operator $T^{(k)}$ is positive for all $k \in \N$.
\end{deff}

\begin{deff}
Let $\mathcal{A}$ be a $^*$-subalgebra of $\mathcal{B}(\hi)$.  A one-parameter family of bounded linear operators $(T_t)_{t \geq 0}$, where $T_t:\mathcal{A} \rightarrow \mathcal{A}$, is called a {\bf semigroup} on $\mathcal{A}$ if
\begin{enumerate}[(i)]
\item $T_{t+s}=T_tT_s$ for all $s,t \geq 0$, and

\item $T_0=1$ where $1$ is the identity operator.
\end{enumerate}
\end{deff}

It should be noted that when we mention weak continuity in the definition below, we do so in the Banach space sense.  We do not mean weak operator topology continuity.
\begin{deff}
A {\bf quantum dynamical semigroup (QDS) on $\mathcal{S}_1(\hi)$ } is a semigroup $(T_t)_{t \geq 0}$ on $\mathcal{S}_1(\hi)$  such that
\begin{enumerate}[(i)]
\item for all $A \in \mathcal{S}_1(\hi)$, the map $t \mapsto T_t(A)$ is weakly continuous,
\item for all $t \geq 0$, $T_t$ preserves traces, that is, $tr(T_t(A))=tr(A)$ for all $A \in \mathcal{S}_1(\hi)$, and
\item for all $t \geq 0$, $T_t$ is completely positive.
\end{enumerate}
\end{deff}
Note that by using the above definition for a QDS we are in the Schr{\"o}dinger picture.  For the Heisenberg picture, a QDS is a semigroup on $\mathcal{B}(\hi)$ with slightly different properties than the ones given above.  We now want to discuss the generator of a QDS.

\begin{deff}
Let $(T_t)_{t \geq 0}$ be a QDS on $\mathcal{S}_1(\hi)$.  Let $D(L) \subseteq \mathcal{S}_1(\hi)$ where $A \in D(L)$ provided
$$weak-\lim_{t \rightarrow 0} \frac{T_t(A)-A}{t}$$
exists and, for $A \in D(L)$, we define the {\bf generator} of the semigroup to be the operator $L:D(L) \rightarrow \mathcal{S}_1(\hi)$ which is given by
$$L(A)=weak-\lim_{t \rightarrow 0} \frac{T_t(A)-A}{t}.$$
\end{deff}

\section{The form of the Generator of a QDS of an N-level System}

In general, the generator of a QDS on $\mathcal{S}_1(\hi)$ is not necessarily bounded.  Here, however, we are interested in the generator of a QDS which describes an N-level system.  Mathematically, this means our Hilbert space $\hi$ is $\C^N$.  In this case, $\mathcal{S}_1(\hi)=\mathcal{B}(\hi)=M_N(\C)$ and further $\mathcal{B}(\mathcal{S}_1(\hi))=M_N(M_N(\C))=M_{N^2}(\C)$.  In this case we also have that the map $t \mapsto T_t$ is norm-continuous (that is, $\|T_t-T_{t_0} \| \rightarrow 0$ as $t \rightarrow t_0$) which implies that the generator $L$ is bounded,
$$L= \| \cdot \| - \lim_{t \rightarrow 0} \frac{T_t-1}{t}$$
and the semigroup is given by
$$T_t=e^{tL}=\sum_{n=0}^{\infty} \frac{(tL)^n}{n!}=\lim_{n\rightarrow \infty}\left(1-\frac{t}{n}L\right)^{-n} .$$
We now want to find the form of the generator of a QDS for an N-level system.  We start with a proposition.

\begin{prop}\label{proposition2.2}
Let $L$ be the generator of a QDS $(T_t)_{t \geq 0}$ on $M_N(\C)$.  Then, for all $A \in M_N(\C)$, we have that 
\begin{enumerate}[(i)]
\item $L(A^*)=L(A)^*$, and

\item $tr(L(A))=0$.
\end{enumerate}
\end{prop}
\begin{proof}
Let $A \in M_N(\C)$.  For any $t \geq 0$, we have that $T_t(A^*)=T_t(A)^*$ since $T_t \geq 0$ and so
$$L(A^*)=\lim_{t \rightarrow 0} \frac{T_t(A^*)-A^*}{t} = \lim_{t \rightarrow 0} \left( \frac{T_t(A)-A}{t} \right)^* =L(A)^* .$$
Further, since $tr(T_t(A))=tr(A)$ for all $t \geq 0$, we have
$$tr(L(A))=\lim_{t \rightarrow 0} \frac{tr(T_t(A))-tr(A)}{t}=\lim_{t \rightarrow 0}\frac{tr(A)-tr(A)}{t}=0 .$$
This completes the proof.
\end{proof}

Note that $M_N(\C)=\mathcal{S}_2(\C^N)$, where $\mathcal{S}_2(\C^N)$ is the space of Hilbert-Schmidt operators on $\C^N$, and $\mathcal{S}_2(\C^N)$ is a Hilbert space. Hence, it makes sense to talk about an orthonormal basis for $M_N(\C)$ using the Hilbert-Schmidt inner product.

\begin{lem}\label{lemma2.3}
Let $(F_{\alpha})_{\alpha =1}^{N^2}$ be an orthonormal basis for $M_N(\C)$.  Then, for any $A \in M_N(\C)$, we have
$$\sum_{\alpha =1}^{N^2}F_{\alpha}AF_{\alpha}^*=tr(A) \cdot 1 .$$
\end{lem}

\begin{proof}
First, we would like to prove that if $(E_{\beta})_{\beta =1}^{N^2}$ is another orthonormal basis of $M_N(\C)$ then 
\begin{equation}\label{equation7}
\sum_{\alpha =1}^{N^2}F_{\alpha}AF_{\alpha}^*=\sum_{\beta =1}^{N^2}E_{\beta}AE_{\beta}^*.
\end{equation}
Since $(E_{\beta})_{\beta =1}^{N^2}$ is an orthonormal basis for $M_N(\C)$, we have that $F_{\alpha}=\sum_{\beta =1}^{N^2} \la F_{\alpha},E_{\beta} \ra E_{\beta}$ for all $\alpha =1, \dots , N^2$.  Hence,
\begin{align*}
\sum_{\alpha =1}^{N^2}F_{\alpha}AF_{\alpha}^* & = \sum_{\alpha =1}^{N^2}\left(\sum_{\beta =1}^{N^2} \la F_{\alpha},E_{\beta} \ra E_{\beta} \right)A \left( \sum_{\gamma =1}^{N^2} \overline{ \la F_{\alpha}, E_{\gamma} \ra }E_{\gamma}^* \right)
\\ & = \sum_{\beta =1}^{N^2}\sum_{\gamma =1}^{N^2} \left\la \sum_{\alpha =1}^{N^2}\la E_{\gamma}, F_{\alpha} \ra F_{\alpha}, E_{\beta} \right\ra E_{\beta}AE_{\gamma}^*
\\ & = \sum_{\beta =1}^{N^2}\sum_{\gamma =1}^{N^2} \la E_{\gamma},E_{\beta} \ra E_{\beta}AE_{\gamma}^*
\\ & = \sum_{\beta =1}^{N^2}E_{\beta}AE_{\beta}^*
\end{align*}
which proves \eqref{equation7}.

  Let $A \in M_N(\C)$ and let $(A_{ij})_{i,j=1}^N \subseteq \C$ so that $A=[A_{ij}]_{i,j=1}^N$.  Then, for any $x \in \C^N$, we have
$$E_{ij}AE_{ji}x=\la x,e_j \ra \la Ae_i,e_i \ra e_j = A_{ii}(e_j^* \otimes e_j)x=A_{ii}E_{jj}x$$
and so $E_{ij}AE_{ji}=A_{ii}E_{jj}$.  Then,
$$\sum_{i,j=1}^NE_{ij}AE_{ij}^*=\sum_{i,j=1}^NE_{ij}AE_{ji}=\sum_{i,j=1}^NA_{ii}E_{jj}=\left(\sum_{j=1}^NE_{jj} \right)\left(\sum_{i=1}^NA_{ii} \right)=tr(A)\cdot 1$$
which completes the proof.
\end{proof}
We now proceed with another lemma.

\begin{lem}\label{lemma2.4}
Let $(F_{\alpha})_{\alpha =1}^{N^2}$ be an orthonormal basis for $M_N(\C)$.  Then, for any $L \in M_N(M_N(\C))$, there exists $(c_{\alpha , \beta})_{\alpha , \beta =1}^{N^2} \subseteq \C$ so that
$$L(A)=\sum_{\alpha , \beta =1}^{N^2}c_{\alpha , \beta}F_{\alpha}AF_{\beta}^*.$$
Further, if $L(A^*)=L(A)^*$, then $\overline{c_{\alpha , \beta}}=c_{\beta , \alpha}$ for all $\alpha , \beta =1, \dots ,N^2$.
\end{lem}

\begin{proof}
Let $(F_{\alpha})_{\alpha =1}^{N^2}$ be an orthonormal basis for $M_N(\C)$.   For all $\alpha , \beta =1, \dots , N^2$, define $\Gamma_{\alpha , \beta}:M_N(M_N(\C)) \rightarrow M_N(M_N(\C))$ by $\Gamma_{\alpha , \beta}(A)=F_{\alpha}AF_{\beta}^*$ for all $A \in M_N(\C)$.

{\bf Claim:}  The matrices $(\Gamma_{\alpha , \beta})_{\alpha , \beta =1}^{N^2}$ form an orthonormal basis for $M_N(M_N(\C))$.

Indeed, let $\alpha , \beta , u,v \in \{1,2, \dots ,N^2\}$.  Then, 
\begin{align*}
 \la \Gamma_{\alpha , \beta} , \Gamma_{u,v} \ra & =\sum_{k=1}^{N^2} \la \Gamma_{\alpha , \beta}F_k, \Gamma_{u,v}F_k \ra 
\\ &=    \sum_{k=1}^{N^2}tr(F_{\alpha}F_kF_{\beta}^*F_vF_k^*F_u^*)
\\ &=tr\left(F_{\alpha}\left(\sum_{k=1}^{N^2}F_kF_{\beta}^*F_vF_k^*\right)F_u^*\right)
\end{align*}
and, by Lemma \ref{lemma2.3}, we have that 
\begin{align*}
tr\left(F_{\alpha}\left(\sum_{k=1}^{N^2}F_kF_{\beta}^*F_vF_k^*\right)F_u^*\right) & = tr(F_{\alpha}(tr(F_{\beta}^*F_v)\cdot 1)F_u^*)
\\ & = \la F_v,F_{\beta} \ra \la F_{\alpha}, F_u \ra
\\ & =\left\{ \begin{array}{cc}
0 & \text{ if } v \neq \beta \text{ or } \alpha \neq u \\
1 & \text{ if } v=\beta \text{ and } \alpha =u \\
\end{array} \right. .
\end{align*}
So $( \Gamma_{\alpha , \beta})_{\alpha , \beta =1}^{N^2}$ forms an orthonormal set for $M_N(M_N(\C))$ and since $M_N(M_N(\C))$ has dimension $N^4$, we have that $( \Gamma_{\alpha , \beta})_{\alpha , \beta =1}^{N^2}$ is an orthonormal basis for $M_N(M_N(\C))$.  This completes the proof of the claim.

Since $( \Gamma_{\alpha , \beta})_{\alpha , \beta =1}^{N^2}$ is an orthonormal basis for $M_N(M_N(\C))$ and $L \in M_N(M_N(\C))$, there exists $(c_{\alpha , \beta})_{\alpha , \beta =1}^{N^2} \subseteq \C$ so that
$$L(A)= \sum_{\alpha , \beta =1}^{N^2}c_{\alpha ,\beta}F_{\alpha}AF_{\beta}^*, \quad \text{ for all } A \in M_N(\C).$$
Further, if $L(A^*)=L(A)^*$ then 
$$ \sum_{\alpha , \beta =1}^{N^2}c_{\alpha ,\beta}F_{\alpha}A^*F_{\beta}^*= \sum_{\alpha , \beta =1}^{N^2}\overline{c_{\alpha ,\beta}}F_{\beta}AF_{\alpha}^*=\sum_{\alpha , \beta =1}^{N^2}\overline{c_{\beta , \alpha}}F_{\alpha}A^*F_{\beta}^*$$
where we get the last equality by switching the indices.  Then we have that
$$\sum_{\alpha , \beta =1}^{N^2}(c_{\alpha , \beta } - \overline{c_{\beta , \alpha}}) \Gamma_{\alpha, \beta}(A)=0, \quad \text{ for all } A \in M_N(\C)$$
and since $(\Gamma_{\alpha , \beta})_{\alpha , \beta =1}^{N^2}$ is an orthonormal basis for $M_N(M_N( \C))$, we have that $c_{\alpha , \beta}=\overline{c_{\beta , \alpha}}$ for all $\alpha , \beta =1, \dots ,N^2$.  This completes the proof.
\end{proof}

The last lemma gives a form for any $T \in M_N(M_N(\C))$.  The next gives a form for any $T \in M_N(M_N(\C))$ such that $T(A^*)=T(A)^*$ for all $A \in M_N(\C)$.  For example, $T$ has this property if it is positive or, as Proposition \ref{proposition2.2} shows, if it is the generator of a QDS for an N-level system.

\begin{prop}\label{proposition2.5}
Let $L \in M_N(M_N(\C))$ such that $L(A^*)=L(A)^*$ for all $A \in M_N(\C)$.  Then there exists $K \in M_N(\C)$, $(G_p)_{p=1}^{N^2-1}\subseteq M_N(\C)$, and $(\lambda_p)_{p=1}^{N^2} \subseteq \R$ such that 
$$L(A)=AK^*+KA+\sum_{p=1}^{N^2-1}\lambda_pG_pAG_p^*, \quad \text{ for all } A \in M_N(\C)$$
where $tr(K)=\frac{1}{2}\lambda_{N^2}$, $tr(G_p)=0$ for all $p=1,\dots,N^2-1$, and $tr(G_p^*G_q)=\delta_{p,q}$ for all $p,q=1, \dots, N^2-1$.
\end{prop}

\begin{proof}
Let $(F_{\alpha})_{\alpha =1}^{N^2}$ be an orthonormal basis for $M_N(\C)$ where $F_{N^2}=\frac{1}{\sqrt{N}}1$.  By Lemma \ref{lemma2.4}, there exists $(c_{\alpha , \beta})_{\alpha , \beta =1}^{N^2} \subseteq \C$ so that
$$L(A)=\sum_{\alpha , \beta =1}^{N^2} c_{\alpha , \beta}F_{\alpha}AF_{\beta}^*, \quad \text{ for all } A \in M_N(\C)$$
where $c_{\alpha , \beta}=\overline{c_{\beta , \alpha}}$ for all $\alpha , \beta =1, \dots ,N^2$.  Then we have that
\begin{align*}
L(A) & = \sum_{\alpha , \beta =1}^{N^2} c_{\alpha , \beta}F_{\alpha}AF_{\beta}^*
\\ & =\sum_{\alpha , \beta =1}^{N^2-1} c_{\alpha , \beta}F_{\alpha}AF_{\beta}^* + \frac{1}{\sqrt{N}}\sum_{\alpha =1}^{N^2-1}c_{\alpha , N^2}F_{\alpha}A+\frac{1}{\sqrt{N}}\sum_{\beta =1}^{N^2-1}c_{N^2,\beta}AF_{\beta}^*+\frac{1}{N}c_{N^2,N^2}A
\\ & = \sum_{\alpha , \beta =1}^{N^2-1} c_{\alpha , \beta}F_{\alpha}AF_{\beta}^*+\left(\frac{1}{\sqrt{N}}\sum_{\alpha =1}^{N^2-1}c_{\alpha , N^2}F_{\alpha}\right)A+A \left(\frac{1}{\sqrt{N}}\sum_{\alpha =1}^{N^2-1}c_{\alpha , N^2}F_{\alpha} \right)^*+\frac{1}{N}c_{N^2,N^2}A
\\ & =  \sum_{\alpha , \beta =1}^{N^2-1} c_{\alpha , \beta}F_{\alpha}AF_{\beta}^* +KA+AK^* \stepcounter{equation}\tag{\theequation}\label{equation4}
\end{align*}
where
$$K=\frac{1}{\sqrt{N}}\sum_{\alpha =1}^{N^2-1}c_{\alpha , N^2}F_{\alpha}+ \frac{1}{2N}c_{N^2,N^2}1 .$$
Since $c_{\alpha , \beta}=\overline{c_{\beta , \alpha}}$ for all $\alpha , \beta =1, \dots , N^2$, the matrix $[c_{\alpha , \beta}]_{\alpha , \beta =1}^{N^2-1}$ is self-adjoint so, by the Spectral Theorem for self-adjoint matrices, there exists a unitary matrix $U=[u_{\alpha , \beta}]_{\alpha , \beta =1}^{N^2-1}$ and a diagonal matrix $D=[\lambda_{\alpha , \beta}]_{\alpha , \beta =1}^{N^2-1}$, that is $\lambda_{\alpha , \beta}=0$ for all $\alpha \neq \beta$, where $\lambda_{\alpha}=\lambda_{\alpha , \alpha} \in \R$, such that $[c_{\alpha , \beta}]_{\alpha , \beta =1}^{N^2-1}=U^*DU$.  Simple matrix multiplication will then give that
$$c_{\alpha , \beta}=\sum_{p=1}^{N^2-1}\lambda_p\overline{u_{p,\alpha}}u_{p, \beta}, \quad \text{ for all } \alpha , \beta =1, \dots, N^2-1.$$
Then, from Equation \eqref{equation4}, we have that, for any $A \in M_N(\C)$,
\begin{align*}
L(A) &= \sum_{\alpha , \beta =1}^{N^2-1}c_{\alpha , \beta}F_{\alpha}AF_{\beta}^*+KA+AK^*
\\ & =  \sum_{\alpha , \beta =1}^{N^2-1}\sum_{p=1}^{N^2-1}\lambda_p\overline{u_{p,\alpha}}u_{p, \beta}F_{\alpha}AF_{\beta}^*+KA+AK^*
\\ & = \sum_{p=1}^{N^2-1}\lambda_p\left(\sum_{\alpha =1}^{N^2-1} \overline{u_{p,\alpha}}F_{\alpha} \right)A \left( \sum_{\beta =1}^{N^2-1}\overline{u_{p,\beta}}F_{\beta} \right)^* +KA+AK^*
\\ & = \sum_{p=1}^{N^2-1}\lambda_pG_pAG_p^*+KA+AK^*
\end{align*}
where
$$G_p=\sum_{\alpha =1}^{N^2-1}\overline{u_{p, \alpha}}F_{\alpha}, \quad \text{ for all } p=1, \dots , N^2-1 .$$
Since $(F_{\alpha})_{\alpha =1}^{N^2-1} \cup \{ \frac{1}{\sqrt{N}}1 \}$ is orthonormal, we have that $tr(F_{\alpha})=0$ for all $\alpha =1, \dots , N^2-1$ and so it is straightforward to see that $tr(K)=\frac{1}{2}\lambda_{N^2}$, where we define $\lambda_{N^2}=c_{N^2,N^2}$, and $tr(G_p)=0$ for all $p=1, \dots ,N^2-1$.  To see that $tr(G_p^*G_q)=\delta_{p,q}$ for all $p,q=1, \dots , N^2-1$, we calculate the trace using the equation for $G_p$ given above to obtain
$$tr(G_p^*G_q)=\sum_{\alpha , \beta =1}^{N^2-1}u_{p, \alpha}\overline{u_{q, \beta}}tr(F_{\alpha}^*F_{\beta})=\sum_{\alpha =1}^{N^2-1}u_{p, \alpha}\overline{u_{q, \alpha}}$$
and $\sum_{\alpha =1}^{N^2-1}u_{p, \alpha}\overline{u_{q, \alpha}}= \delta_{p,q}$ since it is the inner product of the $p$th column of $U^*$ with the $q$th column of $U^*$ and $U^*$ is unitary.  This completes the proof.
\end{proof}

When we say a semigroup $(T_t)_{t \geq 0}$ is {\bf $^*$-preserving} we mean that $T_t(A^*)=T_t(A)^*$ for all $A \in M_N(\C)$ and all $t \geq 0$.  The form given in the last proposition actually characterizes the generators of all uniformly continuous, $^*$-preserving semigroups on $M_N(\C)$ as the next corollary shows.

\begin{cor}
Let $(T_t)_{t \geq 0}$ be a uniformly continuous semigroup on $M_N(\C)$ and let $L$ be its generator.  Then $(T_t)_{t \geq 0}$ is $^*$-preserving if and only if there exists $K \in M_N(\C)$, $(G_p)_{p=1}^{N^2-1}\subseteq M_N(\C)$, and $(\lambda_p)_{p=1}^{N^2} \subseteq \R$ such that 
$$L(A)=AK^*+KA+\sum_{p=1}^{N^2-1}\lambda_pG_pAG_p^*, \quad \text{ for all } A \in M_N(\C)$$
where $tr(K)=\frac{1}{2}\lambda_{N^2}$, $tr(G_p)=0$ for all $p=1,\dots,N^2-1$, and $tr(G_p^*G_q)=\delta_{p,q}$ for all $p,q=1, \dots, N^2-1$.
\end{cor}

\begin{proof}
For the forward direction, we can see from the proof of Proposition \ref{proposition2.2} that since $(T_t)_{t \geq 0}$ preserves stars we have that $L$ preserves stars and so we can apply Proposition \ref{proposition2.5}.

For the backwards direction, it is straightforward to see from the form of $L$ that is given we have that $L(A^*)=L(A)^*$ for all $A \in M_N(\C)$ and so
\begin{equation}\label{equationcor}
T_t(A^*)=e^{tL}(A^*)=\sum_{k=0}^{\infty}\frac{t^k}{k!}L^k(A^*)=\sum_{k=0}^{\infty}\frac{t^k}{k!}L^k(A)^*=\left(\sum_{k=0}^{\infty}\frac{t^k}{k!}L^k(A)\right)^*=e^{tL}(A)^*=T_t(A)^*.
\end{equation}
This completes the proof.
\end{proof}

As we previously mentioned, if $L$ is the generator of a quantum dynamical semigroup then we know, from Proposition \ref{proposition2.2}, that $L(A^*)=L(A)^*$ and so we are able to write $L$ in the form given in Proposition \ref{proposition2.5}.  Of course, the fact that $L(A^*)=L(A)^*$ does not characterize the generator of a QDS so we would like to incorporate other properties of generators.  The next proposition further assumes the semigroup generated by $L$ preserves traces, that is, $tr(T_t(A))=tr(A)$ for all $A \in M_N(\C)$ and all $t \geq 0$, which implies $tr(L(A))=0$ for all $A \in M_N(\C)$.

\begin{prop}\label{proposition2.6}
Let $L \in M_N(M_N(\C))$ such that $L(A^*)=L(A)^*$ for all $A \in M_N(\C)$ and $tr(L(A))=0$ for all $A \in M_N(\C)$.  Then there exist $H \in M_N(\C)$, $(G_p)_{p=1}^{N^2-1} \subseteq M_N(\C)$, and $(\lambda_p)_{p=1}^{N^2-1} \subseteq \R$ such that
$$L(A)=-i[H,A]+\frac{1}{2} \sum_{p=1}^{N^2-1}\lambda_p([G_p,AG_p^*]+[G_pA,G_p^*]), \quad \text{ for all } A \in M_N(\C)$$
where $H$ is self-adjoint, $tr(H)=0$, $tr(G_p)=0$ for all $p=1, \dots, N^2-1$, and $tr(G_p^*G_q)=\delta_{p,q}$ for all $p,q=1, \dots , N^2-1$.
\end{prop}

\begin{proof}
Since $L \in M_N(M_N(\C))$ such that $L(A^*)=L(A)^*$ for all $A \in M_N(\C)$ we have, by Proposition \ref{proposition2.5}, that there exists $K \in M_N(\C)$, $(G_p)_{p=1}^{N^2-1} \subseteq M_N(\C)$, and $(\lambda_p)_{p=1}^{N^2} \subseteq \R$ such that
\begin{equation}\label{equation2.6}
L(A)=AK^*+KA+\sum_{p=1}^{N^2-1}\lambda_pG_pAG_p^*, \quad \text{ for all }A \in M_N(\C)
\end{equation}
where $tr(K)=\frac{1}{2}\lambda_{N^2}$, $tr(G_p)=0$ for all $p=1, \dots ,N^2-1$, and $tr(G_p^*G_q)=\delta_{p,q}$ for all $p,q=1, \dots ,N^2-1$.
Since $tr(L(A))=0$ for all $A \in M_N(\C)$ we have that
\begin{align*}
0 & = \sum_{p=1}^{N^2-1}\lambda_p tr(G_pAG_p^*)+tr(AK^*)+tr(KA)
\\ & = tr \left( \left( \sum_{p=1}^{N^2-1}\lambda_pG_p^*G_p+K^*+K \right)A \right) .
\end{align*}
This is true for all $A \in M_N(\C)$ and so
$$0= \sum_{p=1}^{N^2-1} \lambda_pG_p^*G_p+K^*+K ,$$
that is,
$$\Re{(K)}=-\frac{1}{2}\sum_{p=1}^{N^2-1}\lambda_pG_p^*G_p$$
where $\Re{(K)}$ denotes the real part of $K$.
From Equation \eqref{equation2.6} and by substituting in the above equation, we obtain
\begin{align*}
L(A) & = \sum_{p=1}^{N^2-1}\lambda_pG_pAG_p^*+A(\Re{(K)}+i\Im{(K)})^*+(\Re{(K)}+i\Im{(K)})A
\\ & = \sum_{p=1}^{N^2-1}\lambda_pG_pAG_p^*-\frac{1}{2}\sum_{p=1}^{N^2-1}\lambda_pAG_p^*G_p-\frac{1}{2}\sum_{p=1}^{N^2-1}\lambda_pG_p^*G_pA-iA\Im{(K)}+i\Im{(K)}A
\\ & = i[\Im{(K)},A]+\frac{1}{2}\sum_{p=1}^{N^2-1}\lambda_p([G_p,AG_p^*]+[G_pA,G_p^*])
\end{align*}
where $\Im{(K)}$ is the imaginary part of $K$.  Set $H=-\Im{(K)}$.  Clearly $H$ is self-adjoint since the imaginary part of any operator is self-adjoint.  We know $tr(K)=\frac{1}{2}\lambda_{N^2}$ and, if we recall the equation for $K$ from Proposition \ref{proposition2.5}, we see that $tr(K^*)=\frac{1}{2}\lambda_{N^2}$ as well.  Hence, $tr(H)=tr(-\Im{(K)})=0$.  This completes the proof.
\end{proof}

The form given in the last proposition characterizes the generators of all uniformly continuous, $^*$-preserving, and trace preserving semigroups as the next corollary shows.

\begin{cor}
Let $(T_t)_{t \geq 0}$ be a uniformly continuous semigroup on $M_N(\C)$ with generator $L$.  Then $(T_t)_{t \geq 0}$ is $^*$-preserving and trace preserving if and only if there exist $H \in M_N(\C)$, $(G_p)_{p=1}^{N^2-1} \subseteq M_N(\C)$, and $(\lambda_p)_{p=1}^{N^2-1} \subseteq \R$ such that
$$L(A)=-i[H,A]+\frac{1}{2} \sum_{p=1}^{N^2-1}\lambda_p([G_p,AG_p^*]+[G_pA,G_p^*]), \quad \text{ for all } A \in M_N(\C)$$
where $H$ is self-adjoint, $tr(H)=0$, $tr(G_p)=0$ for all $p=1, \dots, N^2-1$, and $tr(G_p^*G_q)=\delta_{p,q}$ for all $p,q=1, \dots , N^2-1$.
\end{cor}

\begin{proof}
For the forward direction, the fact that $(T_t)_{t \geq 0}$ preserves stars and traces implies $L(A^*)=L(A)^*$ and $tr(L(A))=0$ for all $A \in M_N(\C)$ and so we can simply apply Proposition \ref{proposition2.6}.
For the backwards direction, suppose $L$ is of the form given in the statement of the corollary.  Then, since $L$ is a sum of commutators, it is easy to see that $tr(L(A))=0$ for all $A \in M_N(\C)$.  Then, for $A \in M_N(\C)$,
$$0=tr(L(T_t(A)))=\lim_{h \rightarrow 0}\frac{1}{h} \left( tr(T_{h+t}(A))-tr(T_t(A)) \right) = \frac{d}{dt}tr(T_t(A))$$
and so the map $t \mapsto tr(T_t(A))$ is constant.  But, $tr(T_0(A))=tr(A)$, hence $tr(T_t(A))=tr(A)$ for all $t \geq 0$.  Further, from the form of $L$ it is straightforward to check that $L(A^*)=L(A)^*$ for all $A \in M_N(\C)$ and so Equation \eqref{equationcor} gives that $T_t(A^*)=T_t(A)^*$ for all $A \in M_N(\C)$.
\end{proof}

We now have a form for any matrix $L\in M_N(M_N(\C))$ such that $L(A^*)=L(A)^*$ and $tr(L(A))=0$ for all $A \in M_N(\C)$.  While generators of QDS satisfy both of these requirements, the result above still does not characterize generators of QDS.  The remaining property needed is for $\lambda_p \geq 0$ for all $p=1, \dots , N^2-1$.

\begin{prop}\label{proposition2.8}
Let $L$ be the generator of a QDS on $M_N(\C)$.  Then there exists $H \in M_N(\C)$, $(G_p)_{p=1}^{N^2-1} \subseteq M_N(\C)$, and positive scalars $\lambda_1, \dots , \lambda_{N^2-1}$ such that 
\begin{equation}
L(A)=-i[H,A]+\frac{1}{2} \sum_{p=1}^{N^2-1}\lambda_p([G_p,AG_p^*]+[G_pA,G_p^*]), \quad \text{ for all } A \in M_N(\C)
\end{equation}
where $H$ is self-adjoint, $tr(H)=0$, $tr(G_p)=0$ for all $p=1, \dots, N^2-1$, and $tr(G_p^*G_q)=\delta_{p,q}$ for all $p,q=1, \dots , N^2-1$.
\end{prop}

\begin{proof}
Since $L$ is the generator of a quantum dynamical semigroup, by Proposition \ref{proposition2.2}, $L(A^*)=L(A)^*$ and $tr(L(A))=0$ for all $A \in M_N(\C)$ and so, by Proposition \ref{proposition2.6}, there exist $H \in M_N(\C)$, $(G_p)_{p=1}^{N^2-1} \subseteq M_N(\C)$, and $(\lambda_p)_{p=1}^{N^2-1} \subseteq \R$ such that
$$L(A)=-i[H,A]+\frac{1}{2} \sum_{p=1}^{N^2-1}\lambda_p([G_p,AG_p^*]+[G_pA,G_p^*]), \quad \text{ for all } A \in M_N(\C)$$
where $H$ is self-adjoint, $tr(H)=0$, $tr(G_p)=0$ for all $p=1, \dots, N^2-1$, and $tr(G_p^*G_q)=\delta_{p,q}$ for all $p,q=1, \dots , N^2-1$.  Since $tr(G_p^*G_q)=\delta_{p,q}$ for all $p,q=1, \dots ,N^2-1$ and $tr(G_p\cdot 1)=tr(G_p)=0$ for all $p=1,\dots, N^2-1$, we know $(G_p)_{p=1}^{N^2}$ is an orthonormal basis, where we let $G_{N^2}=\frac{1}{\sqrt{N}}1$.  We proceed with two claims.

{\bf Claim 1:}  Let $\{P_1, \dots ,P_N \}$ be a set of mutually orthogonal self-adjoint projections in $M_N(\C)$.  If $L$ is the generator of a QDS then $tr(P_rL(P_s)) \geq 0$ for all $r,s=1, \dots ,N$ where $r \neq s$.

Let $r,s=1, \dots , N$ where $r \neq s$.  Then 
\begin{equation}\label{equation278}
tr(P_rL(P_s))=\lim_{t \rightarrow 0^+} tr\left( P_r\frac{T_t(P_s)-P_s}{t}\right) =\lim_{t \rightarrow 0^+}\frac{1}{t}tr(P_rT_t(P_s)) .
\end{equation}
Let $e_1, \dots ,e_N$ be an orthonormal basis for $\C^N$ such that $span(e_k)_{k=1}^m=Rang(P_r)$ and $span(e_k)_{k=m+1}^N=Rang^{\perp}(P_r)$ for some $1 \leq m \leq N$.  Then, from Equation \eqref{equation278}, we have
\begin{align*}
tr(P_rL(P_s)) & =\lim_{t \rightarrow 0^+}\frac{1}{t}tr(P_rT_t(P_s))
\\ & = \lim_{t \rightarrow 0^+} \frac{1}{t} \sum_{k=1}^N \la P_re_k,T_t(P_s)e_k \ra
\\ & = \lim_{t \rightarrow 0^+} \frac{1}{t} \sum_{k=1}^m \la e_k,T_t(P_s)e_k \ra
\\ & \geq 0
\end{align*}
since $T_t$ is positive, hence $T_t(P_s)$ is positive, and so $\la e_k,T_t(P_s)e_k \ra \geq 0$ for $k=1, \dots, m$ and for all $t \geq 0$.  This completes the proof of the first claim.  

{\bf Claim 2:}  The operator $L^{(N)}$ generates a quantum dynamical semigroup on $M_N(\C) \otimes M_N(\C)$.  

Let $T_{t,N}=e^{tL^{(N)}}$.  We first want to show that $T_{t,N}=T_t^{(N)}$.  To this end, let $\sum_{i,j=1}^nA_{ij}\otimes E_{ij} \in M_N(\C) \otimes M_N(\C)$ where $A_{ij}\in M_N(\C)$ for all $i,j=1, \dots , N$. Then,
\begin{align*}
T_{t,N} \left(\sum_{i,j=1}^nA_{ij}\otimes E_{ij} \right) & = \sum_{k=0}^{\infty}\frac{t^k}{k!}(L^{(N)})^k\left( \sum_{i,j=1}^nA_{ij}\otimes E_{ij} \right)
\\ & = \sum_{k=0}^{\infty}\frac{t^k}{k!} \sum_{i,j=1}^NL^k(A_{ij}) \otimes E_{ij}
\\ & = \sum_{i,j =1}^N \left( \sum_{k=0}^{\infty}\frac{t^k}{k!}L^k(A_{ij}) \right) \otimes E_{ij}
\\ & = \sum_{i,j=1}^NT_t(A_{ij}) \otimes E_{ij}
\\ & = (T_t^{(N)}) \left(\sum_{i,j=1}^nA_{ij}\otimes E_{ij} \right)
\end{align*}
and so $T_{t,N}=T_t^{(N)}$ for all $t \geq 0$.  Then, since $T_t$ is completely positive if and only if $T_t^{(k)}$ is positive for all positive integers $k$, we have that $T_t^{(N)}$ is completely positive.  Further, 
\begin{align*}
tr\left(T_t^{(N)} \left( \sum_{i,j=1}^N A_{ij} \otimes E_{ij} \right) \right) & = tr\left( \sum_{i,j=1}^N T_t(A_{ij}) \otimes E_{ij} \right)
\\ & = \sum_{i=1}^N tr(T_t(A_{ii}))
\\ & = \sum_{i=1}^N tr(A_{ii})
\\ & = tr \left(  \sum_{i,j=1}^N A_{ij} \otimes E_{ij} \right)
\end{align*}
and so $T_t^{(N)}$ preserves traces for all $t \geq 0$.  Therefore, $L^{(N)}$ generates a QDS on $M_N(\C) \otimes M_N(\C)$.  This completes the proof of the second claim.

Now, we would like to apply the result of Claim 1 to the generator $L^{(N)}$ with a carefully chosen family of mutually orthogonal self-adjoint projections in $M_N(\C) \otimes M_N(\C)$ to show that $\lambda_p \geq 0$ for all $p=1, \dots , N^2-1$.  Let
$$R_q=\sum_{i,j=1}^N G_qE_{i,j}G_q^* \otimes E_{i,j} \quad \text{ for } q=1, \dots ,N^2.$$
We first want to show that $\{R_1, \dots , R_{N^2} \}$ is a mutually orthogonal family of self-adjoint projections in $M_N(\C) \otimes M_N(\C)$.  First,
\begin{align*}
R_q^* & = \sum_{i,j=1}^N(G_qE_{i,j}G_q^*)^* \otimes E_{i,j}^*
\\ & = \sum_{i,j=1}^NG_qE_{j,i}G_q^* \otimes E_{j,i}
\\ & =\sum_{i,j=1}^N G_qE_{i,j}G_q^* \otimes E_{i,j} && \text{by switching indices}
\\ & =R_q
\end{align*}
and so $R_q$ is self-adjoint for all $q=1, \dots N^2$.  Also,
\begin{align*}
\sum_{q=1}^{N^2} \sum_{i,j=1}^N G_qE_{i,j}G_q^* \otimes E_{i,j} & = \sum_{i,j=1}^N \left( \sum_{q=1}^{N^2} G_qE_{i,j}G_q^* \right) \otimes E_{i,j}
\\ & =\sum_{i,j=1}^Ntr(E_{i,j})1_N \otimes E_{i,j} && \text{by Lemma \ref{lemma2.3} }
\\ & = \sum_{i=1}^N1_N \otimes E_{i,i}
\\ & =1_{N^2}.
\end{align*}
Further, for $q,s=1, \dots ,N^2$,
\begin{align*}
R_qR_s & =\sum_{i,j=1}^N \left( \sum_{k=1}^N G_qE_{i,k}G_q^*G_sE_{k,j}G_s^* \right) \otimes E_{i,j}
\\ & = \sum_{i,j=1}^N tr(G_q^*G_s)G_qE_{i,j}G_s^* \otimes E_{i,j}
\\ & = \delta_{qs}R_q.
\end{align*}
Hence, $\{R_1, \dots ,R_{N^2} \}$ is a mutually orthogonal family of self-adjoint projections in $M_N(\C) \otimes M_N(\C)$.  Then, by applying Claim 1, we have that for all $q=1, \dots , N^2-1$,
\begin{align*}
0 & \leq Ntr(R_q(L^{(N)})R_{N^2})
\\ & =Ntr \left[ R_q \left( \frac{1}{N} \sum_{i,j=1}^N L(E_{ij}) \otimes E_{ij} \right) \right]
\\ & =tr \left[ \sum_{i,j=1}^N \sum_{k=1}^N G_qE_{ik}G_q^*L(E_{kj}) \otimes E_{ij} \right]
\\ & = tr \left[ \sum_{i=1}^N \sum_{k=1}^N G_qE_{ik}G_q^*L(E_{ki}) \otimes E_{ii} \right]
\\ & = \sum_{i=1}^N \sum_{k=1}^N tr(G_qE_{ik}G_q^*L(E_{ki}))
\\ & = \sum_{i=1}^N \sum_{k=1}^N tr \left[G_qE_{ik}G_q^* \left( E_{ki}K^*+KE_{ik} + \sum_{p=1}^{N^2-1} \lambda_pG_pE_{ki}G_p^* \right) \right]
\\ & =tr(G_q^*)tr(G_qK^*)+tr(G_q)tr(G_q^*K)+\sum_{p=1}^{N^2-1} \lambda_p tr(G_q^*G_p)tr(G_qG_p^*)
\\ & = \lambda_q
\end{align*}
and therefore $\lambda_p \geq 0$ for all $q=1, \dots ,N^2-1$.  This completes the proof.
\end{proof}
The resulting properties of the generator given in Proposition \ref{proposition2.8} turn out to be necessary and sufficient conditions for a matrix $L$ to be the generator of a QDS.  The details are given in the following theorem.

\begin{thm}
A linear operator $L: M_N(\C) \rightarrow M_N(\C)$ is the generator of a quantum dynamical semigroup if and only if there exist $H \in M_N(\C)$, $(G_p)_{p=1}^{N^2-1} \subseteq M_N(\C)$, and positive scalars $\lambda_1, \dots ,\lambda_{N^2-1}$ such that
\begin{equation}\label{maineq}
L(A)=-i[H,A]+\frac{1}{2}\sum_{p=1}^{N^2-1} \lambda_p ([G_p,AG_p^*]+[G_pA,G_p^*]), \quad \text{ for all } A \in M_N(\C)
\end{equation}
where $H$ is self-adjoint, $tr(H)=0$, $tr(G_p)=0$ for all $p=1, \dots , N^2-1$, and $tr(G_p^*G_q)=\delta_{pq}$ for all $p,q=1, \dots ,N^2-1$.
\end{thm} 	

\begin{proof}
The "only if" part of the statement is given by Proposition \ref{proposition2.8}.  For the "if" part of the statement, let $L$ be a linear operator of the form given in Equation \eqref{maineq}.  Define $T_t:M_N(\C) \rightarrow M_N(\C)$ by $T_t=e^{tL}$ for all $t \geq 0$.  Then $L$ is the generator of the semigroup $(T_t)_{t \geq 0}$.  Once we show that that for all $t \geq 0$, the operator $T_t$ is completely positive and preserves traces then the proof will be complete.  To this end, note that $tr(L(A))=0$ for all $A \in M_N(\C)$ since $L(A)$ is a sum of commutators.  Then, for $A \in M_N(\C)$,
$$0=tr(L(T_t(A)))=\lim_{h \rightarrow 0}\frac{1}{h} \left( tr(T_{h+t}(A))-tr(T_t(A)) \right) = \frac{d}{dt}tr(T_t(A))$$
and so the map $t \mapsto tr(T_t(A))$ is constant.  But, $tr(T_0(A))=tr(A)$, hence $tr(T_t(A))=tr(A)$ for all $t \geq 0$.

Now, we want to show that $T_t$ is completely positive for all $t \geq 0$.  In order to do so, we will show that $T_t^{(k)}$ is positive for all $k \in \N$.  Let $k \in \N$ and define $L^{(k)}:M_N(\C) \otimes M_k(\C) \rightarrow M_N(\C) \otimes M_k(\C)$ by $L^{(k)}(A \otimes E_{ij})=L(A) \otimes E_{ij}$.  Further, define $H^{[k]}, G_p^{[k]} \in M_N(\C) \otimes M_k(\C)$ by $H^{[k]}=\sum_{i=1}^kH \otimes E_{ii}$ and $G_p^{[k]}=\sum_{i=1}^kG_p \otimes E_{ii}$.  Then, by Equation \eqref{maineq}, we have that
$$L^{(k)}(x)=-i[H^{[k]},x]+\frac{1}{2}\sum_{p=1}^{N^2-1} \lambda_p ([G_p^{[k]},x(G_p^{[k]})^*]+[G_p^{[k]}x,(G_p^{[k]})^*]), \quad \text{ for all } x \in M_N(\C) \otimes M_k(\C).$$

\noindent{\bf Claim 1:} For all $x \in M_N(\C) \otimes M_k(\C)$,
$$L^{(k)}(x^*x)-x^*L^{(k)}(x)-L^{(k)}(x^*)x+x^*L^{(k)}(1)x \geq 0.$$

For convenience, let $\phi^{[k]}=\frac{1}{2}\sum_{p=1}^{N^2-1} \lambda_p ([G_p^{[k]},x(G_p^{[k]})^*]+[G_p^{[k]}x,(G_p^{[k]})^*])$ so that $L^{(k)}(x)=-i[H^{[k]},x]+\phi^{[k]}(x)$.  Then, with a bit of algebra, it is straightforward to see that
\begin{align*}
& L^{(k)}(x^*x)-x^*L^{(k)}(x)-L^{(k)}(x^*)x+x^*L^{(k)}(1)x 
\\ & = \phi^{[k]}(x^*x)-x^*\phi^{[k]}(x)-\phi^{[k]}(x^*)x+x^*\phi^{[k]}(1)x
\\ & = \sum_{p=1}^{N^2-1} \lambda_p \left(G_p^{[k]}x^*x(G_p^{[k]})^*-x^*G_p^{[k]}x(G_p^{[k]})^*-G_p^{[k]}x^*(G_p^{[k]})^*x+x^*G_p^{[k]}(G_p^{[k]})^*x \right)
\\ & =\sum_{p=1}^{N^2-1}\lambda_p \left(x(G_p^{[k]})^*-(G_p^{[k]})^*x \right)^*\left(x(G_p^{[k]})^*-(G_p^{[k]})^*x \right)
\\ & \geq 0
\end{align*}
which proves the claim.  We now proceed with another claim.

\noindent{\bf Claim 2:}  For all $a,b \in M_N(\C) \otimes M_k (\C)$ such that $ab=0$, we have that $b^*L^{(k)}(a^*a)b \geq 0$.

By Claim 1, we have that
$$L^{(k)}(a^*a) \geq a^*L^{(k)}(a)+L^{(k)}(a^*)a-a^*L^{(k)}(1)a$$
and so
$$b^*L^{(k)}(a^*a)b \geq b^*a^*L^{(k)}(a)b+b^*L^{(k)}(a^*)ab-b^*a^*L^{(k)}(1)ab=0$$
since $ab=0$.  This proves Claim 2.  We now have one final claim.

\noindent{\bf Claim 3:}  For all $\lambda \in \R$ such that $\lambda > \|L^{(k)} \|$, we have that $(1-\lambda^{-1}L^{(k)})^{-1} \geq 0$.

To this end, we want to show that for all $a \in M_N(\C) \otimes M_k(\C)$ such that $(1-\lambda^{-1}L^{(k)})(a) \geq 0$ we have that $a \geq 0$.  Our first goal is to establish that it is enough to show that for all self-adjoint $a \in M_N(\C) \otimes M_k(\C)$ such that $(1-\lambda^{-1}L^{(k)})(a) \geq 0$ we have that $a \geq 0$.  So, let $a \in M_N(\C) \otimes M_k(\C)$ and suppose  $(1-\lambda^{-1}L^{(k)})(a) \geq 0$.  Then,
\begin{align*}
0 & \leq (1-\lambda^{-1}L^{(k)})(a)
\\ & =(1-\lambda^{-1}L^{(k)})(\Re{(a)}+i\Im{(a)})
\\ & =(1-\lambda^{-1}L^{(k)})(\Re{(a)})+i(1-\lambda^{-1}L^{(k)})(\Im{(a)})
\end{align*}
and so, $(1-\lambda^{-1}L^{(k)})(\Im{(a)})=0$ and $(1-\lambda^{-1}L^{(k)})(\Re{(a)})\geq 0$.  Then, if the statement is true for self-adjoint operators, we have that $\Re{(a)} \geq 0$.  Also, since $(1-\lambda^{-1}L^{(k)})$ is invertible and $(1-\lambda^{-1}L^{(k)})(\Im{(a)})=0$, we have that $\Im{(a)}=0$ and so $a=\Re{(a)} \geq 0$.  So, to prove the claim, it suffices to show that for all self-adjoint $a \in M_N(\C) \otimes M_k(\C)$ such that $(1-\lambda^{-1}L^{(k)})(a) \geq 0$ we have that $a \geq 0$.  Let $a \in M_N(\C) \otimes M_k(\C)$ be self-adjoint and suppose that  $(1-\lambda^{-1}L^{(k)})(a) \geq 0$.  Let $a=x-y$ where $x$ and $y$ are the positive and negative parts of $a$, respectively.  Then $\sqrt{x}y=0$, so by Claim 2, $0 \leq y^*L^{(k)}(\sqrt{x}^*\sqrt{x})y=yL^{(k)}(x)y$.  Further, by Claim 3, we have that $0 \leq (1-\lambda^{-1}L^{(k)})(a)$, and so
\begin{align*}
0 & \leq y\left  (1-\lambda^{-1}L^{(k)})(a) \right)y
\\ & = yay-\lambda^{-1}yL^{(k)}(a)y
\\ & =y(x-y)y-\lambda^{-1}yL^{(k)}(x-y)y
\\ & =-y^3-\lambda^{-1}yL^{(k)}(x)y+\lambda^{-1}yL^{(k)}(y)y
\\ & \leq -y^3+\lambda^{-1}yL^{(k)}(y)y
\end{align*}
since $yL^{(k)}(x)y \geq 0$.  Hence, $0 \leq y^3 \leq \lambda^{-1}yL^{(k)}(y)y$ and thus
$$ \|y\|^3=\|y^3\| \leq \lambda^{-1} \|yL^{(k)}(y)y \| \leq \lambda^{-1} \|L^{(k)} \| \|y \|^3$$
and so $y=0$ since $\|L^{(k)} \| < \lambda$.  Therefore, $a=x \geq 0$.  This completes the proof of Claim 3.
Now, let $t \geq 0$.  Then, for large enough $n$ we have, by Claim 3, that $(1-\frac{t}{n}L^{(k)})^{-n} \geq 0$ and so
$$T_t^{(k)}=e^{tL^{(k)}}=\lim_{n \rightarrow \infty}\left(1-\frac{t}{n}L^{(k)}\right)^{-n}\geq 0.$$
This completes the proof.
\end{proof}

\end{spacing}

\end{document}